\documentclass[5p,times]{elsarticle}
\usepackage{hyperref}

\journal{Information Science} 

\usepackage{balance}
\usepackage{listings}
\usepackage{color}

\definecolor{dkgreen}{rgb}{0,0.6,0}
\definecolor{gray}{rgb}{0.5,0.5,0.5}
\definecolor{mauve}{rgb}{0.58,0,0.82}

\makeatletter
\def\lst@makecaption{%
  \def\@captype{table}%
  \@makecaption
}
\makeatother

\newtheorem{definition}{Definition}

\graphicspath{{img/}}

\usepackage{url}
\usepackage{multirow}

\usepackage{amssymb}
\usepackage{amsmath}
\newtheorem{example}{Example}  
\newtheorem{strategy}{Strategy}    
\newtheorem{proof}{Proof}    
\newtheorem{lemma}{Lemma}  
\newtheorem{theorem}{Theorem}   
\usepackage{algorithm}
\usepackage{algorithmic}

\usepackage{algorithmic}
\usepackage{array}

\usepackage[caption=false,font=footnotesize,labelfont=sf,textfont=sf]{subfig}

\begin{document}

\begin{frontmatter}

\title{Correlated Utility-based Pattern Mining}

\author{Wensheng Gan$ ^{1,5} $,
	Jerry Chun-Wei Lin$ ^{1,2}$*,
	Han-Chieh Chao$ ^{3} $,
	Hamido Fujita$ ^{4} $ and
	Philip S. Yu$ ^{5} $}
	\address{$ ^{1} $Harbin Institute of Technology (Shenzhen), Shenzhen, China} 
	\address{$ ^{2} $Western Norway University of Applied Sciences, Bergen, Norway}
	\address{$ ^{3} $National Dong Hwa University, Hualien, Taiwan}
	\address{$ ^{4} $Iwate Prefectural University, Iwate, Japan}
	\address{$ ^{5} $University of Illinois at Chicago, Chicago, USA}
	\address{Email: wsgan001@gmail.com, jerrylin@ieee.org, hcc@ndhu.edu.tw, HFujita-799@acm.org, psyu@uic.edu}

\cortext[cor1]{Corresponding author. Email: jerrylin@ieee.org.} 

\begin{abstract}

In the field of data mining and analytics, the utility theory from Economic can bring benefits in many real-life applications. In recent decade, a new research field called utility-oriented mining has already  attracted great attention. Previous studies have, however, the limitation that they rarely consider the inherent correlation of items among patterns. Consider the purchase behaviors of consumer, a high-utility group of products (w.r.t. multi-products) may contain several very high-utility products with some low-utility products. However, it is considered as a valuable pattern even if this behavior/pattern may be not highly correlated, or even happen by chance. In this paper, in light of these challenges, we propose an efficient utility mining approach namely non-redundant \textbf{\underline{Co}}rrelated high-\textbf{\underline{U}}tility \textbf{\underline{P}}attern \textbf{\underline{M}}iner (CoUPM) by taking positive correlation and profitable value into account. The derived patterns with high utility and strong positive correlation can lead to more insightful availability than those patterns only have high profitable values. The utility-list structure is revised and applied to store necessary information of both correlation and utility. Several pruning strategies are further developed to improve the efficiency for discovering the desired patterns. Experimental results show that the non-redundant correlated high-utility patterns have more effectiveness than some other kinds of interesting patterns. Moreover, efficiency of the proposed CoUPM algorithm significantly outperforms the state-of-the-art algorithm. 

\end{abstract}


\end{frontmatter}


\section{Introduction}

In many real-world applications, data mining \cite{chen1996data,gan2017data} turns a large collection of data into knowledge, and one of the common tasks of data mining is pattern mining \cite{agrawal1994fast,han2004mining,2gan2017mining}. For instance, to analyze the users' click-stream or purchase behavior that contains auxiliary valuable with hidden information, pattern mining algorithms \cite{agrawal1994fast,han2004mining} can be applied to identify embedded patterns and useful knowledge. In the past decades, numerous data mining frameworks and approaches, e.g., frequent pattern mining (FPM) \cite{agrawal1994fast,han2004mining,3gan2018survey} and association rule mining (ARM) \cite{agrawal1994fast}, have been extensively studied. FPM extracts frequent patterns, and ARM aim at mining association rules. Besides, FPM is considered as the first step of ARM and more challenging. In general, these desired patterns represent interesting relationships among objects or patterns in different types of databases. Mining of insightful patterns has been successfully applied in many real-world applications. However, most of these pattern mining algorithms \cite{agrawal1994fast,han2004mining,2gan2017mining} mainly measure the interestingness of patterns based on the co-occurrence frequencies of patterns. Other implicit factors in data such as the weight, interest, risk or profit of patterns are not effectively utilized. Besides, all objects are considered to have equal importance, hence the objects/patterns that are real important to users may not be found effectively.

Therefore, some researchers are interested in incorporating both subjective measure (e.g., risk, interest and utility) and objective measures (e.g., correlation, frequency and confidence) for mining valuable patterns, such as itemsets and association rules. Among them, one utility-oriented data mining framework called high-utility pattern mining (HUPM) \cite{liu2005two,ahmed2009efficient,tseng2013efficient,liu2012mining} was proposed. Inspired by the utility theory \cite{marshall2005principles}, HUPM incorporates some useful factors, e.g., quantity, unit profit, and other useful factors, to identify the patterns which can bring valuable profits for retailers or managers in business. In general, the \textit{utility} can also be other user-specified subjective measure, e.g., risk, interest, significance, satisfaction, and usefulness, etc. The concept of HUPM has been extended to utility mining \cite{liu2005two,ahmed2009efficient,tseng2013efficient,liu2012mining} and it serves as a critical role in data science. Up to now, utility mining has become an important branch of data analytics, which aims at utilizing the auxiliary information from data. It has been widely utilized to discover valuable information  and hidden knowledge in recent decade since utility mining can bring more benefits in many real-life applications. Many studies of utility mining focus on developing the efficient algorithms, such as Two-Phase~\cite{liu2005two}, IHUP \cite{ahmed2009efficient}, UP-growth \cite{tseng2010up}, UP-growth+ \cite{tseng2013efficient}, HUI-Miner \cite{liu2012mining},  FHM \cite{fournier2014fhm}, HUP-Miner \cite{krishnamoorthy2015pruning}, and EFIM \cite{zida2017efim}. At the same time, many studies focus on the effectiveness for mining utility-oriented patterns. For instance, mining high-utility patterns from uncertain data \cite{lin2016efficient}, dynamic data \cite{lin2015fast,2gan2018survey}, and big data \cite{lin2015mining}. 

Utility mining has been extensively studied and successfully applied in many real-world applications \cite{2gan2018survey}. However, the existing studies of utility mining are mainly focused on the identification of high-utility patterns themselves, and thus the hidden correlation among the derived patterns is still limited. In other words, they ignore the inherent correlation of objects/items inside the patterns. This problem may easily lead to the identification of high-utility patterns with false negatives and false positives. Therefore, an important limitation of current utility mining algorithms is that a huge amount of patterns may be discovered while most of them contain many weakly correlated items.  For example, it is common that retail stores cross-sell some products/items to improve the total revenue. Some products are usually sold with discount or free gifts to stimulate the sale of other related products/items. As shown in Figure \ref{fig:AmazonExample}\footnote{\url{https://www.amazon.com}}, many products are bought together for cross-selling in Amazon. This example explains the reasons why correlation is an important factor, especially in utility mining. The really strongly correlated products (or purchase behaviors) are more useful for cross-selling; otherwise, those meaningless, redundant or non-discriminative patterns may be misleading for recommendation.  Hence, it is a critical issue to address the effectiveness problem for discovering positively correlated and high-utility patterns based on the utility and correlation measures.

\begin{figure}[htbp]
	\centering
	\includegraphics[scale=0.53]{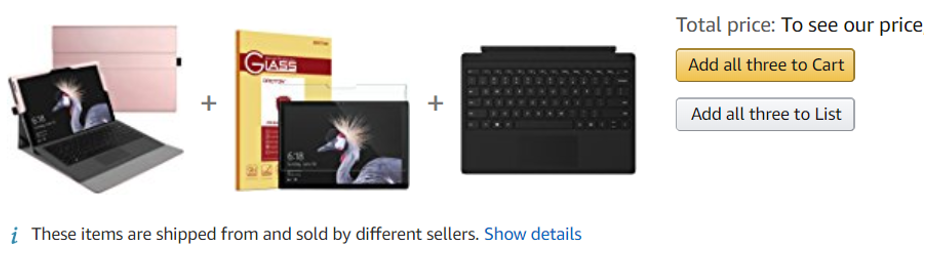}
	\captionsetup{justification=centering}
	\caption{Bought together products in Amazon} 
	\label{fig:AmazonExample}	
\end{figure}

In the past few decades, some well-known correlation measures, e.g., the support \cite{agrawal1993mining,geng2006interestingness}, confidence \cite{geng2006interestingness}, all-confidence \cite{omiecinski2003alternative}, frequency affinity \cite{ahmed2011framework}, and coherence \cite{omiecinski2003alternative}, have been studied in data mining. In the field of utility mining, the HUIPM \cite{ahmed2011framework} and FDHUP \cite{lin2017fdhup} algorithms were proposed to discover high-utility interesting patterns (HUIPs) with strong frequency affinity. The concept of affinity utility is introduced in HUIPM \cite{ahmed2011framework}. However, the tree-based HUIPM algorithm is time-consuming and may lead to the problem of combinatorial explosion. The faster FDHUP algorithm \cite{lin2017fdhup} utilizes two compact data structures and three pruning strategies to efficiently discover discriminative HUIPs.  However, the co-occurrence frequency instead of inherent correlation is measured as the correlation factor in HUIPM \cite{ahmed2011framework} and FDHUP \cite{lin2017fdhup}. Recently, a projection-based approach namely CoHUIM \cite{gan2018extracting} was developed to discover correlated high-utility patterns with consideration of the inherent correlation among items inside a pattern. It adopts a measure called \textit{Kulc} \cite{kulczynski1928pflanzenassoziationen,wu2010re}, which has the \textit{null} (transaction)-invariant property, as the correlation factor. The discovered patterns have strong positive inherent correlation, and they can bring real benefits to utility mining. However, the projection-based CoHUIM  may encounter the efficiency problem, and may cause a lot of memory consumption since it relies on the candidate generation-and-test mechanism.

In light of the above challenges, we propose an efficient utility mining framework, namely non-redundant \textbf{\underline{Co}}rrelated high-\textbf{\underline{U}}tility \textbf{\underline{P}}attern \textbf{\underline{M}}iner (CoUPM) with the consideration of strong positive correlation and utility theory. CoUPM can not only extract non-redundant strongly correlated and profitable patterns, but also achieve better efficiency. We evaluate the effectiveness of the proposed CoUPM based on the correlation measure \textit{Kulc}. For comparison, we take the well-known traditional HUPM model and the state-of-the-art CoHUIM algorithm into account to compare the designed algorithm for the correlated utility-based pattern mining problem. The major contributions of this paper are summarized as follows.

\begin{itemize}
	\item We adopt correlated significance as a key criterion for evaluating the high-utility patterns in the HUPM problem. Understanding such correlation can provide useful insights on the discovered results, and this makes utility mining with a higher effective performance than the existing HUPM models. The utility factor and relations among items/objects are taken into account for pattern evaluation.  
	
	\item We design an efficient CoUPM algorithm for mining correlated and high-utility patterns from quantitative databases in one-phase. The revised utility-list structure is used to store the compact information of potential patterns from the processed database. This approach is able to early filter a large amount of unpromising patterns, and return the significant patterns in the mining process.
	
	\item We develop several pruning techniques in a depth-first search manner, which consist of the utilization of correlation property and utility property. Therefore, CoUPM can quickly discover a set of highly correlated and high-utility patterns.
	
	\item Extensive experiments on both real and synthetic datasets show that the proposed one-phase CoUPM algorithm has better effectiveness and efficiency than the existing algorithms.

\end{itemize}

The rest of this paper is organized as follows. Some related works of utility mining are briefly reviewed in Section \ref{sec:relatedwork}. The key preliminaries and problem statement are given in Section \ref{sec:preliminaries}.  Details of the proposed CoUPM algorithm are described in Section \ref{sec:algorithm}. The evaluation of effectiveness and efficiency of CoUPM are provided in Section \ref{sec:experiments}.  Finally, conclusion and future work are drawn in Section \ref{sec:conclusion}.

\section{Related Work}
\label{sec:relatedwork}
This research work is related to the studies in support-based pattern mining, utility-based pattern mining, and the development of pattern mining with consideration of affinity/correlation.

\subsection{Utility-based Data Mining}

In the past few decades, many pattern mining frameworks and algorithms have been developed and applied to various real-life applications. Most of these studies use support \cite{agrawal1994fast,agrawal1993mining} and confidence \cite{agrawal1993mining} to identify interesting patterns, e.g., frequent patterns \cite{agrawal1994fast,han2004mining,agrawal1993mining}. These studies  measure the interestingness of patterns mainly based on the co-occurrence frequency \cite{agrawal1994fast,han2004mining}. Therefore, many interesting and high profitable patterns may not be found. To address these problems, a new data mining framework named utility-oriented pattern mining \cite{tseng2013efficient,yao2006mining,gan2018survey} is proposed. It aims at discovering the high-utility patterns rather than the support/confidence-based patterns.  Utility mining considers the quantity and unit profit of objects/items, as well as other implicit factors. In the past decade, the problem of high-utility pattern mining (HUPM) has been extensively studied, such as Two-Phase \cite{liu2005two}, IHUP \cite{ahmed2009efficient}, UP-growth \cite{tseng2010up}, UP-growth+ \cite{tseng2013efficient}, and HUI-Miner \cite{liu2012mining}. There are mainly four categories of the existing HUPM algorithms, including Apriori-like, tree-based, utility-list-based, and hybrid approaches. The well-known Apriori-like approach for HUPM is the Two-Phase \cite{liu2005two} algorithm which utilizes the transaction-weighted utilization (TWU) model \cite{liu2005two}. Inspired by FP-growth \cite{han2004mining}, some tree-based algorithms are proposed to mining high-utility patterns, such as IHUP \cite{ahmed2009efficient}, UP-growth \cite{tseng2010up}, UP-growth+ \cite{tseng2013efficient}, and HUP-tree \cite{lin2011effective}. All of them outperform the Apriori-like algorithms. Liu et al. then introduced the HUI-Miner \cite{liu2012mining} by utilizing the utility-list structure and a new concept called remaining utility. Recently, many other utility-list based approaches have been developed, such as  FHM \cite{fournier2014fhm}, HUP-Miner \cite{krishnamoorthy2015pruning}, and EFIM \cite{zida2017efim}.

The above mentioned HUPM algorithms focus on improving the mining efficiency, however, the effectiveness of utility mining task is also quite important. For example, how to develop different and flexible models to address the utility mining task in different types of data, constraints and applications are very necessary and challenging. Up to now, some studies that focus on the effectiveness of utility mining have been extensively developed, such as HUPM on uncertain data \cite{lin2016efficient} or dynamic data \cite{lin2015fast,2gan2018survey}. Lin et al. proposed a series of models to extract high-utility patterns from uncertain data \cite{lin2016efficient,lin2017efficiently}, temporal data \cite{lin2015efficient}, and dynamic data \cite{lin2015fast,2gan2018survey,3lin2016fast}. Based on the new concept of average utility \cite{hong2011effective}, Wu et al. introduced a new upper bound for mining high average utility patterns \cite{wu2018tub}. Besides, several evolutionary computation approaches (e.g., HUIM-BPSO \cite{lin2017binary} and ACO-based HUIM-ACS \cite{wu2017aco})  are proposed to discover high-utility patterns. Tseng et al. introduced the concise representation \cite{tseng2015efficient} and top-$k$ issue  \cite{tseng2016efficient} for HUPM. Different from the itemset-based models, other advanced models were extensively studied, including the association rule-based \cite{mai2017lattice}, sequence-based \cite{lan2014applying}, and episode-based \cite{lin2015discovering} utility mining models. Recently, Gan et al. proposed a new utility measure named utility occupancy to address the utility mining problem \cite{gan2019huopm}. An overview of the current development of utility mining was presented recently \cite{2gan2018survey}.

\subsection{Affinity/Correlation Pattern Mining}

In the data mining literature, several association measures for association mining and analytics have been studied, such as confidence \cite{agrawal1993mining}, lift \cite{brin1997beyond}, and the cosine measure \cite{wu2012cosine}. Association analysis may generate many rules, while many of them are not useful or meaningful for decision-making. Different from association analysis, some studies have been explored for mining affinity patterns or correlation patterns. Omiecinski et al. first proposed three interesting measures for pattern mining called any-confidence \cite{omiecinski2003alternative}, all-confidence \cite{omiecinski2003alternative}, and bond \cite{omiecinski2003alternative}. To find strong affinity patterns which  may contain low-support items, Kim et al. first introduced  hyperclique pattern and hyperclique (\textit{h})-confidence \cite{kim2004ccmine}. The \textit{h}-confidence is equivalent to the all-confidence. Wu et al. found that the degree of expectation-based correlation is highly influenced by the number of null transactions \cite{wu2010re}. Thus, most of the existing measures, e.g., all-confidence \cite{omiecinski2003alternative}, bond \cite{omiecinski2003alternative}, cosine \cite{wu2012cosine}, are not suitable to evaluate correlation in large database that contains many and unstable null transactions. Due to the \textit{null} (transaction)-invariant property, the correlation measure in Kulczynsky \cite{kulczynski1928pflanzenassoziationen,wu2010re} is independent of the dataset size. Besides, some other measures for the study of correlation have been proposed \cite{geng2006interestingness}. Different from the traditional data mining approaches which ignore the correlation among extracting results, the derived affinity/correlation patterns can return more insightful knowledge for decision-making.

\subsection{Comparative Analysis with Previous Works}

As mentioned before, the inherent correlation of items inside the patterns has not been considered in most of the HUPM algorithms yet. In the area of utility-oriented pattern mining, only few studies concern the utility and correlation together to derive the desired patterns. For instance, the HUIPM \cite{ahmed2011framework} and FDHUP \cite{lin2017fdhup} algorithms consider both frequency-affinity and utility as the two key measures to derive the desired patterns. However, the co-occurrence frequency of transactions is regarded as the correlation factor. Recently, Fournier-Viger et al. \cite{fournier2016mining} introduced a FCHM model to extract correlated high-utility itemsets (CoHUIs). In the framework of FCHM, the bond measure \cite{omiecinski2003alternative} is used to evaluate the correlation value of items among a pattern. Moreover, the projection-based CoHUIM \cite{gan2018extracting}  algorithm has presented to take the correlation measure $-$ \textit{Kulc}, which has null (transaction)-invariant property, into account for mining the interesting patterns. However, it may encounter the efficiency problem and may easily cause a lot of memory consumption. The reason is that CoHUIM firstly generates the complete correlated high-utility upper-bound itemsets (CHUUBIs) by recursively processing the projection, which  uses the upper bound TWU \cite{liu2005two} and \textit{Kulc} measure. It then calculates actual utilities for all candidates in CHUUBIs to discover the final CoHUIs. In this paper, the proposed CoUPM method utilizes the revised utility-list structure and several powerful pruning strategies to significantly improve the mining efficiency.


\section{Preliminaries and Problem Formulation}
\label{sec:preliminaries}

In this section, we first introduce some basic preliminaries of utility mining, and then discuss the differences between the addressed problem in this paper and the existing tasks. Finally, we provide a normal problem formulation of correlated high-utility pattern mining.

\subsection{Database with Utility Factor}

Note that we use the concept of \textit{utility} to present the revenue for sellers. In the following contents, let $X$ = $\{i_1,$ $i_2,$ $ \dots, i_k\}$ denote a combination/group of patterns/products, and $X$ is called a $k$-itemset. In general, a unit profit of $X$ is associated to the \textit{cost price} minus \textit{sell price}. As mentioned before, the \textit{utility} concept can be regarded as other user-specified subjective measure, e.g., risk, interestingness, satisfaction, usefulness, etc. According to the utility theory \cite{marshall2005principles}, we have the following concepts and formulation.

\begin{example}
	Consider an e-commerce database as shown in Table \ref{figDatabase}, it is used as a running example in this paper. Similar to the e-commerce database provided by RecSys Challenge 2015\footnote{\url{https://recsys.acm.org/recsys15/challenge/}}, this example database contains five purchase records (e.g., $T_1, T_2, \dots, T_5$) with auxiliary information. Behavior $T_1$ is occurred in timestamp ``07/12 10:05:30", and contains products $\{a\}$, $\{b\}$, and $\{e\}$ with a purchase quantity of 3, 1 and 2, respectively. Table \ref{figUnitProfits} indicates that the unit profit (also called external profit) of these three products is $\{$a$: \$3\}$, $\{$b$: \$1\}$, and $\{$e$: \$10\}$, respectively. Note that the unit profit of each product is pre-defined by user/seller. In the addressed problem, this table is called \textit{profit-table}. 
\end{example}

\begin{table}
	\centering
	\caption{An e-commerce database}
	\begin{tabular}{ccll} \hline
		\textbf{Tid}  & \textbf{User} & \textbf{Timestamp}	&  \textbf{Purchase record}   \\ \hline
		$T_1$ &	 $U_1$  &  07/12  10:05:30 & $(a,3)(b,1)(e,2)$  \\ 
		$T_2$ &	 $U_2$  &  07/12  10:11:10 &	$(a,2)(b,3)(c,1)(d,1)$ \\
		$T_3$ &	 $U_3$  &  07/12  10:15:48	& $(a,1)(d,3)(e,2)$  \\
		$T_4$ &	 $U_4$  &  07/12  10:18:00 & 	$(a,1)(b,5)(c,2)(d,1)(e,1)$  \\ 
		$T_5$ &	 $U_5$  &  07/12  10:25:20 & 	$(a,2)(b,3)(e,3)$  \\ \hline
	\end{tabular}
	\label{figDatabase}
\end{table}

\subsection{Preliminaries of Utility Mining}

Given a \emph{quantitative database} $D$ such that $D$ = $\{T_1,$ $T_2,$ $ \dots, T_n\}$ contains a set of quantitative transactions. Each transaction $T_c$ with a timestamp is a set of items/records. $T_c \in D$ is a subset of $I$ and has a unique identifier called its \textit{Tid}. Let $I$ be a set of distinct items, $I$ = $\{i_1,$ $i_2,$ $\dots, i_m\}$. Each item $i \in I$ is associated with a positive value $pr(i)$ namely its \emph{unit profit}. For each item $i \in T_c$, a positive number $q(i, T_c)$ is called \emph{occur quantity} of $i$. The utility contribution of a group of products, $X$ = $\{i_1,$ $i_2,$ $\dots, i_j\}$, is related to the total utilities from each $i \in X$ after marketing.

\begin{table}
	\centering
	\caption{External profit value (unit profit)}
	\begin{tabular}{|c|c|c|c|c|c|} \hline
		\textbf{Product} & $a$ & $b$ & $c$ & $d$ & $e$ \\ \hline
		\textbf{Profit} (\$) & 3 & 1 & 7 & 2 & 10 \\ \hline
	\end{tabular}
	\label{figUnitProfits}
\end{table}

\begin{definition}
	\rm The utility of a group of products $X \subseteq I$ in a transaction $T_c$ is $ u(X, T_c)$ = $\sum_{i \in X} {u(i, T_c)}$, where $u(i, T_c)$ is the \emph{utility/profit of a product} $i \in I$ in a transaction $T_c$, and $u(i, T_c)$ can be calculated as $u(i, T_c)$ = $pr(i) \times q(i, T_c)$. Thus, $ u(X, T_c)$  represents the utilities generated by all items $i \in X$ in $T_c$. Consider the entire database, let $u(X)$ denote the total utility of $X$ in $D$, then $ u(X)$ = $\sum_{X\subseteq T_{c}\wedge T_{c}\in D} u(X, T_{c}) $.  
\end{definition}

\begin{definition}
	\rm Given a quantitative database $D$, the transaction utility  of a transaction $ T_{c} $, denoted as $tu(T_{c})$, can be calculated as $tu(T_{c})$ = $ \sum_{i_{j}\in T_{c}}u(i_{j}, T_{c}) $, where $i_j$ is the $j$-th item in $ T_c $. Then the total utility of the entire $D$ is denoted as \textit{TU}, and can be calculated as: \textit{TU} = $\sum_{T_{c} \in D}tu(T_{c}) $.
\end{definition}

\begin{example} 
	In Table \ref{figDatabase}, the utility of $e$ in $T_3$ is $u(e, T_3)$ = $2 \times \$10$ = \$20, and the utility of $\{d,e\}$ in $T_3$ is  $u(\{d,e\}, T_3)$ =  $u(d, T_3)$ + $u(e, T_3)$ = 3 $\times$ \$2 + 2 $\times $ \$10 = \$26. The utility of $\{d,e\}$ in the entire database is $u(\{d,e\})$ = $u(\{d,e\},T_3)$ + $u(\{d,e\},T_4)$ = \$26 + \$12 = \$38. Consider the first transaction in Table \ref{figDatabase}, $ tu(T_{1}) $ = $ u(a, T_{1}) $ + $ u(b, T_{1}) $ + $ u(e, T_{1}) $ = \$9 + \$1 + \$20 = \$30. Then the transaction utilities of \textit{T}$ _{1} $ to \textit{T}$ _{5} $ are respectively calculated as \textit{tu}(\textit{T}$ _{1} $) = \$30, \textit{tu}(\textit{T}$ _{2} $) = \$18, \textit{tu}($T_{3} $) = \$29, \textit{tu}($T_{4}$) = \$34, and \textit{tu}($T_{5}$) = \$39. Thus, the total utility of in Table \ref{figDatabase} is $TU$ = \$30 + \$18 + \$29 + \$34 + \$39 = \$150.
\end{example}

\subsection{Correlation for Data Mining}

As stated in introduction, the current HUPM algorithms have an important limitation that a huge amount of derived patterns may contain many items which are weakly correlated. HUIPM \cite{ahmed2011framework} and FDHUP \cite{lin2017fdhup} used a new measure called frequency affinity to evaluate the correlation of high-utility patterns. The minimum  quantity among all quantities of items inside a pattern in each transaction is used to calculate the affinitive frequency. However, it is not enough to reveal  the real inherent correlation of the desired patterns. In the past,  the Kulczynsky (abbreviated as \textit{Kulc}) measure \cite{omiecinski2003alternative,kulczynski1928pflanzenassoziationen,wu2010re} was widely used to evaluate the inherent correlation of a generalized pattern. Its definition is given as follows.

\begin{definition}
	\label{def_Correlation}
	\rm  The pattern correlation evaluates the strength of the inherent correlation between its items. In general, there are three types of correlations among items in a pattern, including 1) \textit{positive correlation}, 2) \textit{non-correlation}, and 3) \textit{negative correlation}.
\end{definition}

\begin{definition}
	\label{def_kulc}
	\rm The \textit{Kulc} value is an interesting measure to evaluate the correlation between items inside a pattern. According to the previous studies \cite{kulczynski1928pflanzenassoziationen,wu2010re}, the \textit{Kulc} value of a group of patterns $X$ is  defined as follows:
	\begin{equation}
		Kulc(X)= \dfrac{1}{k}\sum _{i_{j}\in X}\dfrac{sup(X)}{sup(i_j)},	\end{equation}	
	where $ i_j $ is the $j$-item in $X$ = $\{i_1,$ $i_2,$ $ \dots, i_k\}$ which totally contains $k$ distinct items. 
\end{definition}

Therefore, the range of \textit{Kulc} value is [0, 1] and it can be easily used to evaluate whether the items in a specific pattern have a \textit{positive correlation} or not. Clearly, the minimum correlation threshold for measuring \textit{Kulc} value can be specified by user. Unlike other existing correlation measures, \textit{Kulc} has the \textit{null} (transaction)-invariant property. Previous studies \cite{kulczynski1928pflanzenassoziationen,wu2010re} have shown that \textit{Kulc} value is more acceptable and suitable than other correlation measures to evaluate the correlation in data mining. The reason is that it is independent by the dataset size. Based on the above definitions, we have the following problem formulation.

\begin{example} 
	Consider the running example in Table \ref{figDatabase}, when the settings of \textit{minUtil} and \textit{minCor} are respectively 20\% and 0.7, the desired CoHUIs are the set of \textit{\textbf{Patterns$_{CoHUI}$}} = \{$\{e\}$, $\{a,b\}$, $\{a,c\}$, $\{a,e\}$, $\{b,e\}$, $\{a,b,e\}$, $\{b,c,d\}$\}, while the set of high-utility patterns derived by the exiting HUPM algorithms are 
	\textit{\textbf{Patterns$_{HUI}$}} = \{$\{e\}$, $\{a,b\}$, $\{a,c\}$, $\{a,e\}$, $\{b,e\}$, $\{d,e\}$, $\{a,b,c\}$, $\{a,b,e\}$, $\{a,c,d\}$, $\{a,d,e\}$, $\{b,c,d\}$, $\{a,b,c,d\}$, $\{a,b,c,e\}$, $\{b,c,d,e\}$, $\{a,b,c,d,e\}$\}.	 It is clearly seen that most of the patterns in \textit{\textbf{Patterns$_{HUI}$}} do not have a positive correlated relationship. For instance, the patterns $\{d,e\}$ and $\{a,b,c\}$ have their \textit{Kulc} values as \textit{Kulc}($\{d,e\}$) $ \approx $ 0.5833 and \textit{Kulc}($\{a,b,c\}$) $ \approx $ 0.6333. What's more, the desired interesting CoHUIs do not have the \textit{downward closure} property \cite{agrawal1994fast}. For example, the pattern $\{a\}$ is not a CoHUI, while its supersets $\{a,b\}$, $\{a,c\}$, $\{a,b,e\}$ are the CoHUIs. Previous studies \cite{liu2005two,ahmed2009efficient,tseng2013efficient,liu2012mining,gan2018survey} have shown that the utility of a pattern may be higher, equal to, or lower than that of its super-pattern and/or sub-pattern. Consequently, many  pruning techniques of search space that rely on the downward  closure property of Apriori \cite{agrawal1994fast} cannot be directly applied to discover CoHUIs. 
\end{example}

As far, we have pointed out the major differences between HUIs and CoHUIs. The models aims at finding different patterns regarding to varied problems. Based on above introduction, the addressed problem in this paper is formulated below.

\subsection{Problem Formulation}

\begin{definition}
	\label{def_10}
	\rm A group of patterns \textit{X} in a quantitative database $ D $ is defined as a strongly correlated high-utility itemset (denoted as CoHUI) if it satisfies the following two criteria: 1) $ u(X) \geq minUtil \times TU $; 2) $ Kulc(X) \geq minCor $. Otherwise, $X$ is not a CoHUI, it may have a low utility or a negative correlation.  
	Here, $ minUtil $ is a minimum utility threshold and $ minCor $ is a minimum positive correlation threshold; both of them can be specified by users' subjective preferences. In this paper, $ minUtil $ is a percentage value with respect to the total utility of a quantitative database. Therefore, the problem of correlated utility-based pattern mining (abbreviated as correlated HUPM) is to discover the complete set of  significant and insightful CoHUIs in the entire database.
\end{definition}

HUPM has shown its powerful potential in many applications and achieved outstanding performance compared with the support/confidence based data mining methods. Based on the utility theory \cite{marshall2005principles} and correlation measure, the importance of utility and relations among items/objects are simultaneously taken into account.  The extracted results of CoHUIs are high corresponding to positive correlation and profitable values.

\section{Proposed One-Phase Algorithm: CoUPM}
\label{sec:algorithm}

In this section, we propose an one-phase CoUPM algorithm to discover useful patterns, which are not only strongly correlated but also high profitable. CoUPM utilizes a vertical data structure named revised utility-list. Moreover, several effective pruning strategies which utilize the correlation and utility factors are applied to prune the search space and reduce memory cost. Details of the revised utility-list, the adopted pruning strategies, and the main procedures of the proposed algorithm are respectively described below.  

\subsection{Properties of the CoHUI}
Most existing studies have been demonstrated that both the  \textit{Kulc} measure \cite{kulczynski1928pflanzenassoziationen,wu2010re} and utility measure \cite{liu2005two} are neither monotonic nor anti-monotonic. In other words, a pattern may have a lower, equal or higher \textit{Kulc} value (or utility value) than that of its subsets. Without holding the anti-monotonicity, the search space of the addressed problem is hard to be efficiently reduced in the mining process.  To hold the \textit{downward closure} property for mining high-utility patterns, a concept called \emph{transaction-weighted utilization} \cite{liu2005two} is commonly used in previous studies. 

\begin{definition}
	\rm{
		Given a database $D$ and a specific pattern $X \subseteq D$, the \emph{transaction-weighted utilization} (\textit{TWU}) \cite{liu2005two} of $X$ is defined as the sum of the total utilities of transactions containing $X$, as shown in the following equation: 
		\begin{equation}
		    TWU(X) = \sum _{ X\subseteq T_{c} \wedge T_{c} \in D}tu(T_{c}). 
		\end{equation} 
	}
\end{definition}

\begin{example}
	Consider two patterns $\{e\}$ and $\{d,e\}$ in the running example, then $TWU(e)$ = $tu(T_1)$ + $tu(T_3)$ + $tu(T_4)$ + $tu(T_5)$  = \$30 + \$29 + \$34 + \$39 = \$132, and $TWU(\{d,e\})$ = $tu(T_3)$ + $tu(T_4)$ = \$29 + \$34 = \$63.
\end{example}

Based on the definition of CoHUI and utility property, the CoHUI does not hold the anti-monotonicity. In other words, a CoHUI may have lower, equal or higher utility value (or \textit{Kulc} value)  than any of its subsets. The \textit{TWU} concept solves the anti-monotonicity problem by overestimating the overall utility of patterns in entire database without missing any high-utility patterns. However, a huge number of low-utility patterns still may be regarded as candidates since \textit{TWU} is a loose upper-bound.

\subsection{Revised Utility-List with Correlation}

In  previous studies, several approaches \cite{liu2012mining,fournier2014fhm,krishnamoorthy2015pruning} use the utility-list \cite{liu2012mining} structure as a component to store and calculate the necessary information. Thanks to the vertical data structure of utility-list \cite{liu2012mining}, these approaches can efficiently discover high-utility patterns without multiple database scans. But the original utility-list only deals with utility value, and does not contain the support and correlation information. The addressed problem needs a more flexible version of calculating scheme to obtain the auxiliary information. In the proposed CoUPM algorithm, we revise the utility-list \cite{liu2012mining} to make it suitable for computing the correlation and utility. Besides, a concept called \textit{remaining utility} \cite{liu2012mining} is applied to obtain the estimated upper bound on utility, which will be presented in next subsection. Inspired by the utility-list \cite{liu2012mining} structure, the revised utility-list structure is defined as follows.

\begin{definition} 
	\label{def_totalOrder}
	\rm Without loss of generality, assume that all items in every transaction are sorted in the lexicographic order.  Let the total order on items is denoted as  $\prec$. 
\end{definition}

\begin{definition}
	\rm Let $ ru(X,T_c) $ denote the remaining utility \cite{liu2012mining}  of a group of items/products $X$ in a transaction $T_c$. Then $ ru(X,T_c) $ is the sum of the  utility values of each item appearing after $X$ in $T_c$ according to the total order $\prec$. Thus, the remaining utility of $X$ does not include the utilities of items in $X$ itself. It can be represented as: 
	\begin{equation}
		ru(X,T_c) = \sum_{ i_j \notin X \wedge X \subseteq T_c \wedge X \prec i_j }u(i_j,T_c).
	\end{equation}
\end{definition}

\begin{definition}
	\rm The \textit{revised utility-list} of a pattern $X$ in a quantitative database $D$ consists of pattern name (\textit{\underline{name}}), support count (\textit{\underline{sup}}), and a set of tuples corresponding to the transactions where $X$ appears (\textit{\underline{tuple}}). Here, $\underline{sup}$ is the related support of $X$ that occurred in the entire database, and it is equal to the number of tuples in this vertical data structure. A tuple is defined as $<$$\underline{tid}, \underline{iu}, \underline{ru}$$>$ for each transaction $T_{c}$ containing $X$. 
	\begin{itemize}
		\item $\underline{tid}$: the transaction identifier of $T_{c}$; 
		\item $\underline{iu}$: the actual utility of $X$ in $T_{c}$, w.r.t. $ u(X,T_c) $;
		\item $\underline{ru}$: the remaining utility of $X$ in $T_{c}$, w.r.t. $ ru(X,T_c) $.

	\end{itemize}
\end{definition}

\begin{figure}[!htbp]
	\centering
	\includegraphics[scale=0.65]{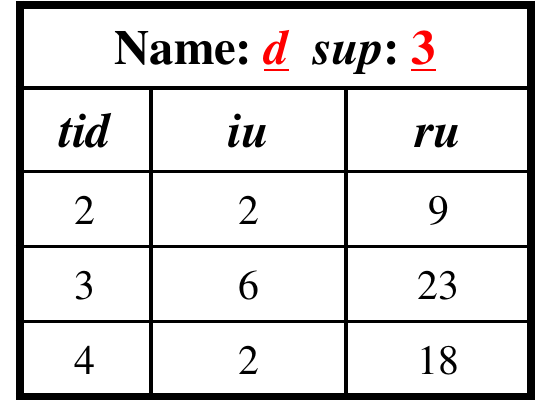}
	\captionsetup{justification=centering}
	\caption{Constructed revised utility-list of 1-itemset ($d$).}
	\label{fig:utilitylistOfD}
\end{figure}

\begin{example} 	
	Consider  $T_4$ and two patterns $\{d\}$ and $\{d,e\}$ in Table \ref{figDatabase}. Assume the total order $\prec$  adopts the support ascending order of all 1-itemsets. Since the support values of all 1-items in Table \ref{figDatabase} are \{$a$:5, $b$:4, $c$:2, $d$:3, $e$:4\}, the total order is $ c \prec d \prec b \prec e \prec a $. Then we have that $ ru(d,T_4)$ = $u(b,T_4)$ + $u(e,T_4)$ + $u(a,T_4) $ = \$5 + \$10 + \$3 = \$18, and $ ru(\{d,e\},T_{4})$ = $u(a,T_4)$ = \$3. Consider the running example and the defined total order $\prec$, the revised utility-list of pattern $(d)$ is \{($T_2$, \$2, \$9), $(T_3$, \$6, \$23), $(T_4$, \$2, \$18)\}, and its total support is 3, as shown in Figure \ref{fig:utilitylistOfD}. 
\end{example}

\begin{figure}[!htbp]
	\centering
	\includegraphics[scale=0.65]{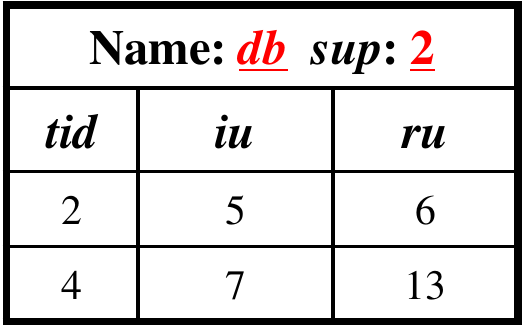}
	\captionsetup{justification=centering}
	\caption{Constructed revised utility-list of 2-itemset ($db$).}
	\label{fig:utilitylistOfBD}
\end{figure}

Unlike the original utility-list \cite{liu2012mining} only deals with utility value, our revised structure can deal with more rich information, including support, correlation and utility. We can perform a single database scan to create the all revised utility-lists of all 1-items in the processed database. After constructing the initial revised utility-list of each 1-item $X \in D$ (denoted as $ X.list$), for any $k$-itemset ($k \geq 2 $), its revised utility-list can be directly constructed using the already built revised utility-lists of its subsets. Note that this operation does not need to scan the database anymore, and the built revised utility-lists fit in main memory. Details of the construction procedure of the revised utility-list are similar to the construction of utility-list, which can be referred to  \cite{liu2012mining}. The difference between them is that after the join operation of two common \textit{tids}, the procedure in the CoUPM algorithm simultaneously updates the support information in the revised utility-list for pattern $ X_{ab}$. This is denoted as $ X_{ab}.list.sup$ ++. For example, the 2-itemset $db$ appears in $T_2$ and $T_4$, and its revised utility-list is constructed based on $ d.list$ and $b.list$, as shown in Figure \ref{fig:utilitylistOfBD}. Note that the construction keeps consistent with  respect to the total order $ c \prec d \prec b \prec e \prec a $.

Note that for optimization, when finding the common \textit{tids} of two itemsets from the two sets of \textit{tids} in the revised utility-lists, we use the binary search to speed up the computational efficiency. For example, we can perform a binary search to find the element with a given \textit{tid} in a target revised utility-list.

\begin{definition}
	\rm Based on the designed revised utility-list, let $X.IU$ and $X.RU$ respectively denote the sum of \textit{utility} values and the sum of \textit{remaining utility} values for a pattern \textit{X} in the constructed  revised utility-list of \textit{X}. According to \cite{lin2016efficient,gan2019beyond}, they can be calculated as follows:
	\begin{equation}X.IU = \sum_{X \subseteq T_c \wedge T_c \in D}u(X,T_c);\end{equation}
	\begin{equation}X.RU = \sum_{X \subseteq T_c \wedge T_c \in D}ru(X,T_c).\end{equation}
\end{definition}

Thus, $X.IU$ of a pattern \textit{X} equals to $u(X)$. Both $X.IU$ and $u(X)$ are the total utility of $X$ in the entire database.

\subsection{Pruning Strategies for Searching CoHUIs}

Similar to previous studies \cite{liu2005two,lin2017fdhup}, the complete search space of the addressed problem can be presented as a Set-enumeration tree \cite{rymon1992search}. This prefix-based tree structure represents all possible itemsets of $I$ where each tree node represents a subset of $I$. It is important to notice that this tree structure is only a conceptual representation and not stored in entirety while performing the mining process. In worst case, this approach may have up to 2$^n$ final itemsets (i.e., all itemsets of the search space with $I$). Without downward closure property, the search space would be huge. To address this limitation, we present a prefix-based depth-first enumeration tree. It means that the node in this tree structure is searched in the depth-first manner.

To speed up performance, the existing CoHUIM algorithm utilizes the \textit{Kulc} measure in non-decreasing order of support count that holds the \textit{sorted downward closure} property \cite{gan2018extracting}. By utilizing the revised utility-list, this \textit{sorted downward closure} property of \textit{Kulc} measure \cite{gan2018extracting} can be applied in the proposed CoUPM algorithm. More importantly, the enumeration of potential patterns may be terminated earlier by \textit{Kulc} value and upper bound on utility. Details of the pruning strategies are described below.

\begin{lemma}[\textbf{Sorted downward closure property of \textit{Kulc}}]
	\label{lemma_Kulc}
	If the items in the set $ \{a_1, a_2, ..., a_k, a_{k+1}\} $ are sorted in support-ascending order, i.e., $ sup(a_1) \leq sup(a_2) \leq ...  \leq sup(a_k) \leq sup(a_{k+1}) $, the \textit{Kulc} measure has the \textit{sorted downward closure} property. That is: $Kulc(a_1... a_k a_{k+1}) \leq Kulc(a_1... a_k) $ \cite{gan2018extracting}.	
\end{lemma}

\begin{proof}	
  \rm	A complete proof can be referred to \cite{gan2018extracting}. 		
\end{proof}

Based on Lemma \ref{lemma_Kulc}, the following \textit{sorted downward closure} property of \textit{Kulc} measure can be held.

\begin{theorem}
	\label{theorem_GDC}
	\rm (\textbf{Anti-monotonicity of \textit{Kulc} with \textit{SDC} property}).	\emph{For any rooted node/itemset in the search space of CoUPM, if a tree node is a correlated pattern, its parent node is also a correlated pattern in $D$. Let $X $ be a $k$-itemset (node), and $X'$ be any of its child nodes (extension, ($k$+1)-itemset). The \textit{Kulc} measure with the \textit{SDC} property is anti-monotonic: $ Kulc(X') \leq Kulc(X) $ always holds.} 
\end{theorem}

\begin{example} 
	Since the support counts of all 1-items are \{$a$:5, $b$:4, $c$:2, $d$:3, $e$:4\}, this set of 1-items is sorted in support-ascending order as \{$ c \prec d \prec b \prec e \prec a $\}. Based on the definition of \textit{Kulc} value (c.f. Formula 1), we can calculate the \textit{Kulc} values of the following patterns, $\{c,d\}$, $\{c,d,b\}$, $\{c,d,b,e\}$ and $\{c,d,b,e,a\}$, as: \textit{Kulc}($\{c,d\}$) $\approx$ 0.833, \textit{Kulc}($\{c,d,b\}$) $\approx$ 0.722, \textit{Kulc}($\{c,d,b,e\}$) $\approx$ 0.333 and \textit{Kulc}($\{c,d,b,e,a\}$) $\approx$ 0.307.
\end{example}

Thus, the \textit{Kulc} measure holds anti-monotonicity if the processed items are sorted in support-ascending order. Note that the total order $\succ$ of items in the Set-enumeration tree \cite{rymon1992search} for the proposed CoUPM algorithm adopts the support-ascending order. Thus, we can utilize the following properties to prune the search space, and the details are described below.

\begin{lemma}[\textbf{Upper bound on utility}] 
	\label{lemma_CDC}
	For any rooted node/itemset $X$ in the search space of CoUPM, the sum of $X.IU$ and $X.RU$ in the revised utility-list of $X$ is always no less than the overall utility of any of its descendant nodes (extensions, denoted as $ X' $). It is an upper bound on utility, such that $X'.IU \leq X.IU+X.RU$.
\end{lemma}

\begin{proof}	
	\rm  A complete proof of this lemma can be referred to \cite{lin2016efficient,gan2019beyond}. 		
\end{proof}

Thus, the sum of the utilities of $ X' $ in \textit{D} would not greater than (\textit{$X.IU$} + \textit{$X.RU$}) of $ X$ in \textit{D}. In other words, (\textit{$X.IU$} + \textit{$X.RU$}) of $ X$ is an upper bound on utility while evaluating the overall utility of a specific pattern.

\begin{example} 
	Consider the running example, assume we perform the depth-first manner in the search space  with the support-ascending order as \{$ c \prec d \prec b \prec e \prec a $\}. When determining the nodes/patterns in the subtree rooted at node $\{c,d,b\}$, we have: $\{c,d,b\}$.$IU$ = \$33, $\{c,d,b\}$.$RU$ = \$19, thus $\{c,d,b\}$.$IU$ + $\{c,d,b\}$.$RU$ = \$33 + \$19 = \$52. All the nodes in the subtree of $\{c,d,b\}$ would not have a utility value higher than the upper bound on utility of node $\{c,d,b\}$. For instance, $\{c,d,b,e\}$.$IU$ = \$31,  $\{c,d,b,a\}$.$IU$ = \$34, $\{c,d,b,e,a\}$.$IU$ = \$42, all are less than \$52. 
\end{example}

\begin{theorem} 
	\label{theorem2}
	(\textbf{Anti-monotonicity of upper-bound on utility}). \emph{For any node/itemset $X$ in the search space of CoUPM, let $X'$ denote any of $X$'s children (extension node). Then the sum of $X.IU$ and $X.RU$ in the revised utility-list of $X$ (equally in the entire database) is always larger than or equal to the sum of $X'.IU$ and $X'.RU$ of $X'$ in the entire database. That is $X'.IU + X'.RU $ $ \leq X.IU + X.RU $ \cite{lin2016efficient,gan2019beyond}.
	}
\end{theorem}

Thus, the sum of total utilities and remaining utilities of $ X $ in $D$ is always larger than or equal to the sum of utilities of its extension in the search space. This upper bound ensures that the \textit{downward closure} property of transitive extensions. Based on the above observations, we can use the following filtering strategies.

\begin{strategy}
	\textbf{Pruning strategy using the SDC property of \textit{Kulc} value, abbreviated as SPK strategy}. \rm Assume the total order $\succ$ of the processed items adopts the support-ascending order. While performing a depth-first search strategy in the search space, if the relative \textit{Kulc} value of any node/itemset $X$ is less than $minCor$, any of its child node is not a CoHUI, and these unpromising patterns can be regarded as irrelevant and pruned directly.
\end{strategy}

\begin{strategy}
	\textbf{Pruning strategy using the upper bound on utility, abbreviated as UBU strategy}. \rm After building the initial revised utility-lists for each 1-itemset, the CoUPM algorithm traverses the search space based on a depth-first search strategy. If the sum of $X.IU$ and $X.RU$ of any node/pattern $X$ is less than $minUtil \times TU $, any of its child node would not be a CoHUI, they can be regarded as irrelevant and pruned directly.
\end{strategy}

To further improve the mining efficiency, the LA-Prune strategy \cite{krishnamoorthy2015pruning} with the upper bound on utility is extended to the proposed algorithm.

\begin{lemma}
	\label{lemma_LAPrune}
	Given two different pattern $X$ and $Y$ ($X \neq Y$), neither $\{X, Y\}$ nor any supersets of $X$ would be a high-utility itemset if \textit{X.IU} + \textit{X.RU} -$ \sum _{X\subseteq T_{q} \wedge T_{q}\subseteq D \wedge Y \nsubseteq T_{q}}((X.iu + X.ru) \leq minUtil) $ \cite{krishnamoorthy2015pruning}.  
\end{lemma}

\begin{strategy}
	\textbf{LA-Prune strategy}. \rm In the search space, let $X$ be a processed pattern (node), and $Y$ be the right sibling node of $X$. If the sum of (\textit{X.IU} + \textit{X.RU}) subtracts the utilities (\textit{X.iu} + \textit{X.ru}) of a set of transactions is less than \textit{minUtil}, the combined pattern $\{X,Y\}$ is not a HUI (also not a CoHUI), and any of child nodes of $X$ would not be a HUI (also not be a CoHUI). During the depth-first search progress, the construction of the revised utility-lists for the children nodes of $X$ is not necessary to be performed.
\end{strategy}

The improved construction procedure is similar to that of revised utility-list. It utilizes the LA-Prune strategy to avoid constructing a huge number of revised utility-lists of the unpromising patterns, as described in Algorithm \ref{AlgorithmConstruct2}.

\renewcommand{\algorithmicrequire}{\textbf{Input:}}
\renewcommand{\algorithmicensure}{\textbf{Output:}}
\begin{algorithm}[!ht]
	\caption{Construction with LA-Prune}
	\label{AlgorithmConstruct2}
	
	\begin{algorithmic}[1]		

		\REQUIRE {$X$: an itemset, $X_{a}$: the extension of $X$ with an item $a$, $X_{b}$: the extension of $X$ with an item $b$ ($a \neq b$)}

		\ENSURE $ X_{ab}$.   
		
		\STATE initialize  $ X_{ab}.list \leftarrow \emptyset $;\
		\STATE set \textit{Utility} = \textit{X.IU} + \textit{X.RU};\
		\FOR {each element/tuple $ E_{a}\in X_{a}.list $}
			\IF {$ \exists E_{a}\in X_{b}.list \wedge E_{a}.tid == E_{b}.tid $}
				\IF{$ X.list \neq \emptyset $}			
					\STATE search for each element $ E\in X.list $ such as $E.tid$ = $E_{a}.tid $;\	
					\STATE $E_{ab} \leftarrow $  $<$$E_{a}.tid, E_{a}.iu$ + $E_{b}.iu - E.iu$, $E_{b}.ru$$>$;\
				\ELSE				
					\STATE $E_{ab} \leftarrow $ $<$$E_{a}.tid$, $E_{a}.iu$ + $E_{b}.iu$, $E_{b}.ru$$>$;\					
				\ENDIF	
		
				\STATE $ X_{ab}.list \leftarrow X_{ab}.list \cup E_{ab} $;\
				\STATE update support information in the revised utility-list for $ X_{ab}$, such as $ X_{ab}.list.sup$ ++.\
			\ELSE {
				\STATE \textit{Utility} = \textit{Utility} - $E_{a}.iu$ - $E_{a}.ru$;\
				\IF {$Utility < minUtil$ }{
					\STATE \textbf{return} \textit{null}.\
				}\ENDIF	
			}\ENDIF	
		\ENDFOR
		
		\STATE \textbf{return} \textit{$ X_{ab} $}
	\end{algorithmic}
\end{algorithm}


\subsection{Main Procedure}

To clarify our methodology, we have illustrated the designed data structure, the key properties of utility and correlation with \textit{Kulc} value, and the upper bound on utility so far. Utilizing the above technologies, the main procedure of the designed CoUPM algorithm is shown in Algorithm \ref{AlgorithmCoUPM}. It takes four parameters as input: 1) an e-commerce quantitative database, $D$; 2) a user-specified profit-table, \textit{ptable}; 3) a minimum positive correlation threshold, \textit{minCor} (0 $\leq minCor \leq $ 1); and 4) a user-specified minimum utility threshold, \textit{minUtil} (0 $\leq minUtil \leq $ 1). When \textit{minCor} is set to 0, it means that CoUPM does not consider the correlation factor.

\renewcommand{\algorithmicrequire}{\textbf{Input:}}
\renewcommand{\algorithmicensure}{\textbf{Output:}}
\begin{algorithm}
	\caption{The CoUPM algorithm}\label{AlgorithmCoUPM}

	\begin{algorithmic}[1]		
		\REQUIRE \textit{D}; \textit{ptable}; \textit{minCor}; \textit{minUtil}.
		\ENSURE \textit{CoHUIs}: the set of correlated high-utility itemsets.      
		
		\STATE scan $D$ once to calculate $ TWU(i) $, construct the \textit{Tidset} of each item $ i\in I$ in $D$, and obtain the \textit{TU};\
		\STATE find all 1-item $i \in I$ such that \textit{TWU}$(i)$ $ \geq$ $ minUtil \times TU $, then put into the set of $I^*$;\
		\STATE use the \textit{Tidset} to sort $I^*$ in the support-ascending order as the total order $\succ$;\ 
		\STATE scan $D$ once again to build the revised utility-list of each itemset $i \in I^*$ using the total order $\succ$;\
		\STATE \textbf{call Search}$\boldmath{(\emptyset, I^*, minCor, minUtil)}$.\		
		\STATE return \textit{CoHUIs}
	\end{algorithmic}
\end{algorithm}

The CoUPM algorithm first scans the database once to calculate $ TWU(i) $ and construct the \textit{Tidset} of each item $ i \in I$ in $D$. The total utility of $D$ is also calculated. Here, the built \textit{Tidset} of all 1-items can be used to sort the items and calculate the \textit{Kulc} value in the later processes. Then all the 1-items which have \textit{TWU}$(i)$ $ \geq$ $ minUtil \times TU$ are put into the set of $I^*$. Thereafter, all patterns do not in the candidate set $I^*$ will be ignored since they cannot be the part of CoHUIs. CoUPM then scans $D$ once again to build the initial revised utility-list of each item $i \in I^*$ using the total order $\succ$.

It is important to notice that the adopted order $\succ$ should be kept consistently after the construction of revised utility-list. In the designed CoUPM algorithm, the support-ascending order is used to hold the \textit{sorted downward closure} property of \textit{Kulc} value. In other words, without using this sorting order, we only can utilize the upper bound on utility w.r.t. Strategy 2 to prune the search space. In the next section of experimental results, we will conduct the proposed CoUPM algorithm with or without using the \textit{sorted downward closure} property of \textit{Kulc} value w.r.t. Strategy 1.

\renewcommand{\algorithmicrequire}{\textbf{Input:}}
\renewcommand{\algorithmicensure}{\textbf{Output:}}
\begin{algorithm}
	\caption{The $Search$ procedure}\label{AlgorithmSEARCH}
	
	\begin{algorithmic}[1]		
		\REQUIRE $X$, $\textit{extensionsOfX}$, $minCor$, $minUtil$.
		\ENSURE the set of \textit{CoHUIs}.

		\FOR {each itemset $ X_{a}\in $ $ \textit{extensionsOfX} $}
			\STATE obtain the $ X_{a}.IU $ and $ X_{a}.RU $ from the built $ X_{a}.list $;
			\STATE calculate the $ Kulc(X_a)$ value using the built $ X_{a}.list $ and \textit{Tidset} of all 1-items;
			\IF {$Kulc(X_a) \geq minCor$ and $ X_{a}.IU\geq minUtil \times TU$}
				\STATE  \textit{CoHUIs} $\leftarrow CoHUIs \cup X_{a} $;\	
		    \ENDIF
		
			\IF {$Kulc(X_a) \geq minCor$ and $ (X_{a}.IU + X_{a}.RU$) $\geq minUtil \times TU$}
				\STATE $ \textit{extensionsOfX}_{a}\leftarrow  \emptyset $;\
				\FOR {each itemset $ X_{b}\in \textit{extensionsOfX} $ such that $ X_{b} $  after  $ X_{a} $}			
					\STATE $ X_{ab}\leftarrow X_{a} \cup X_{b} $;\
					\STATE $ X_{ab}.list\leftarrow $ \textbf{Construct}($X, X_{a}, X_{b}) $;\
					\IF {$ X_{ab}.list \not= \emptyset$}
						\STATE  \textit{extensionsOfX}$_{a}\leftarrow$ \textit{extensionsOfX}$_{a}$ $\cup $ $ X_{ab}.list $;\	
					\ENDIF
				\ENDFOR  		
				\STATE \textbf{call Search}$(X_{a}, \textit{extensionsOfX}_{a}, minCor, minUtil)$.\	
			\ENDIF
		\ENDFOR

		\STATE \textbf{return} \textit{CoHUIs}
	\end{algorithmic}
\end{algorithm}

The $Search$ procedure (as shown in Algorithm \ref{AlgorithmSEARCH}) takes as input: 1) a pattern $X$, 2) extensions of $X$ having the form $X_a$ means that $X_a$ is obtained by appending a pattern $a$ to $X$, 3) $minCor$, and 4) $minUtil$. The \textit{search} procedure operates as follows. It first obtains the $ X_{a}.IU $ and $ X_{a}.RU $ values from the built revised utility-list of $X_a$ (denoted as $ X_{a}.list $) (Line 2). It also calculates the $ Kulc(X_a)$ value using the built $ X_{a}.list $ and \textit{Tidsets} of all 1-items (Line 3 and Eq. (4)). As mentioned previously in Formula 1, the calculation of $ Kulc(X_a)$ value of an itemset $X_a$ is based on all support count of the 1-items containing in this itemset. Notice that here the \textit{Tidsets} of all 1-items just needs to be built once in the first database scan. Since the support count of a special itemset can be easily obtained from its revised utility-list w.r.t. support element, we can quickly calculate this $ Kulc(X_a)$ value.

For each extension $X_a$ of $X$, if the related correlation of $X_a$ is no less than $ minCor$, and the sum of the actual utility of $X_a$ (w.r.t. $ X_{a}.IU $ in  revised utility-list) is no less than \textit{minUtil} $\times TU$, then this pattern is output as a CoHUI (Lines 4 to 5). After that, the designed pruning strategies are used to determine whether the extensions of $X_a$ should be explored or not (Line 6, using Strategy 1 and Strategy 2). This is performed by merging $X_a$ with each extension $X_b$ of $X$ such that $a \succ b$, to form extensions of the form $X_{ab}$ (Lines 9 to 10). The revised utility-list of $X_{ab}$ is then constructed by calling the \textit{Construct} procedure to perform the join operation of the revised utility-lists of $X$, $X_{a}$ and $X_{b}$ (Line 11, details of the construction have been described in Algorithm \ref{AlgorithmConstruct2}). To further filter the unpromising patterns, only the promising patterns with their revised utility-lists would be explored in next extension (Line 11). After all the extensions of the rooted  $X_{a}$ are performed (Line 12), it recursively calls the \textit{Search} procedure with \textit{extensionsOfX}$_{a}$ to continually explore its extension(s) (Line 13).

\section{Experimental Study} \label{sec:experiments}

In this section, we conduct several experiments to demonstrate the effectiveness and efficiency of our proposed model. 

\textbf{Baseline algorithms}. Note that we use one of the traditional HUPM algorithms (e.g., FHM \cite{fournier2014fhm}) and FDHUP to generate the different kinds of discovered results for pattern evaluation, while only the CoHUIM algorithm is compared for efficiency evaluation. The reason is that different kinds of patterns are related to different mining tasks, and they can be used to analyze the effectiveness and usefulness of the CoUPM framework. While the efficiency should be compared with those algorithms which focus on same mining task. Thus, it is unreasonable to evaluate the efficiency by comparing algorithms from different domains.

The CoUPM algorithm is compared with some baseline approaches, including traditional HUPM algorithm which does not consider correlation factor (e.g., HUI-Miner \cite{liu2012mining}, FHM \cite{fournier2014fhm}, and EFIM \cite{zida2017efim}), the frequency-affinity-based FDHUP algorithm \cite{ahmed2011framework}, and the projection-based CoHUIM algorithm \cite{lin2017fdhup}. 

\textbf{Variants of CoUPM algorithm}. Additional to the baseline CoUPM algorithm which only utilizes the Strategy 2, three improved variants, e.g., CoUPM$_{sorted}$ (adopts Strategies 1 and 2), CoUPM$_{LA}$ (adopts Strategies 2 and 3), and CoUPM$_{sorted+LA}$ (adopts Strategies 1, 2 and 3), are used to evaluate the efficiency of the proposed algorithm.

\begin{table*}[htb]
	\fontsize{7.5pt}{9pt}
	\selectfont
	\centering
	\caption{Derived patterns under various parameters}
	\label{table:patterns1}
	\begin{tabular}{|c|c|cccccc|}
		\hline\hline
		\multirow{2}*{\textbf{Dataset}}&
		\multirow{2}*{\textbf{\# Patterns}}
		&\multicolumn{6}{c|}{\textbf{\# patterns under different thresholds}}\\
		\cline{3-8}
		&&$ \alpha_1 $ & $ \alpha_2 $ & $ \alpha_3 $ & $ \alpha_4 $ &  $ \alpha_5 $  & $ \alpha_6 $ \\ \hline

&  \#HUIs  & 93,418	 &  49,821 &  	26,176 &  	14,156 &  	8,364 & 5,365  \\
&  \#DHUIs  & 33,621 &  	16,285 &  	8,712 &  	5,277 &  	3,644 & 2,778  \\
	  &\textbf{\#${P1}$} (\textit{minCor}: 0.01) &	 93,418 & 	49,821 & 	26,176 & 	14,156 & 	8,364 & 	5,365 \\
	 
\textbf{foodmart}  &  \textbf{\#${P2}$} (\textit{minCor}: 0.02)  & 24,857 &  	17,127 &  	12,026 &  	8,557 &	6,290 & 4,683
 	 \\
 &\textbf{\#${P3}$} (\textit{minCor}: 0.03) &  20,224 &  	15,082 &  	11,195 &  	8,247 &  	6,191 & 4,651	 \\
&  \textbf{\#${P4}$} (\textit{minCor}: 0.04) & 6,869 &  	5,675 &  	4,682 &  	3,831 &  	3,222 & 2,764 	 \\
&  \textbf{\#${P5}$} (\textit{minCor}: 0.05) &  4,654 &  	4,092 &  	3,558 &  	3,084 &  	2,712  & 2,405	 \\
&  \textbf{\#${P6}$} (\textit{minCor}: 0.06) & 2,644 & 	2,501 & 	2,344 & 	2,204 & 	2,065 & 	1,944	 \\
\hline

&  \#HUIs  & -  & - &    45,711,058	 &  2,486,972 &  	22,641 &  	15,713   \\
&  \#DHUIs  & 14,539 &  	12,620 &  	11,122 &  	9,873 &  	8,847 &  	7,953  \\
 &  \textbf{\#${P1}$} (\textit{minCor}: 0.10) &  11,510 &  	10,244 &  	9,252 &  	8,474 &  	7,836 &  	7,301   \\
 
\textbf{retail}  &  \textbf{\#${P2}$} (\textit{minCor}: 0.12)  &   10,221 &   	9,224 &   	8,402 &   	7,735 &   	7,170 &   	6,686 	 \\
 &\textbf{\#${P3}$} (\textit{minCor}: 0.14) &   8,734  & 	7,922  & 	7,239  & 	6,676  & 	6,187  & 	5,783 	  \\
&  \textbf{\#${P4}$} (\textit{minCor}: 0.16)  &   7,600  & 	6,929  & 	6,317  & 	5,825  & 	5,402  & 	5,050	  \\
&  \textbf{\#${P5}$} (\textit{minCor}: 0.18) &    7,081  & 	6,469  & 	5,909  & 	5,455  & 	5,063  & 	4,732	  \\
&  \textbf{\#${P6}$} (\textit{minCor}: 0.20)  &   6,823  & 	6,238  & 	5,699  & 	5,260  & 	4,878  & 	4,562	  \\
\hline

&  \#HUIs  & 198,920 &  89,933  & 39,281	 & 16,848 & 	7,141 & 	2,969 \\
&  \#DHUIs  & 0 & 0 & 0 & 0	& 0	 & 0 \\
&  \textbf{\#${P1}$} (\textit{minCor}: 0.74) &  8,717 & 	6,083 & 	4,021 & 	2,493 & 	1,438 & 760 \\
\textbf{chess}  &  \textbf{\#${P2}$} (\textit{minCor}: 0.75)  & 7,330 & 	5,113 & 3,378 & 2,101 & 1,199 & 629 	 \\
&\textbf{\#${P3}$} (\textit{minCor}: 0.76) & 6,062 & 	4,210 & 	2,773 & 	1,704 & 	799	 & 500  \\
&  \textbf{\#${P4}$} (\textit{minCor}: 0.77) &  4,987 & 	3,464 & 	2,282 & 	1,401 & 	799	 & 415  \\
&  \textbf{\#${P5}$} (\textit{minCor}: 0.78) &  4,118 & 	2,856 & 	1,872 & 	1,145 & 	650 & 	333  \\
&  \textbf{\#${P6}$} (\textit{minCor}: 0.79) &  3,316 & 	2,287 & 	1,483 & 	891 & 	503 & 	252 \\
\hline

&  \#HUIs  &  22,121 & 	13,953 & 	7,601 & 		3,420 & 		1,265 & 		356 	 \\
&  \#DHUIs  &  8 & 		2 & 		0 & 		0 & 		0 & 		0 	 \\
 &  \textbf{\#${P1}$} (\textit{minCor}: 0.40)  & 8,732 &	5,320 &	2,723 &	1,145 &	435 &	138 \\
\textbf{mushroom}  &  \textbf{\#${P2}$} (\textit{minCor}: 0.42)  &  7,126 & 	4,422 & 	2,374 & 	1,046 & 	413 & 	129 	 \\
 &\textbf{\#${P3}$} (\textit{minCor}: 0.44) &  5,962 & 	3,778 & 	2,060 & 	928	 & 362 & 	106	 \\
&  \textbf{\#${P4}$} (\textit{minCor}: 0.46)  & 4,783 & 	3,043 & 	1,656 & 	748	 & 280	 & 82	 \\
&  \textbf{\#${P5}$} (\textit{minCor}: 0.48) &  3,595 & 	2,236 & 	1,176 & 	508	 & 195 & 	61 	 \\
&  \textbf{\#${P6}$} (\textit{minCor}: 0.50) &  2,452 & 	1,478 & 	729 & 	301 & 	107 & 	40 	 \\
\hline

&  \#HUIs  &  370,624 & 	167,972 & 	91,529 & 	56,326 & 	37,381 & 	26,385 \\
&  \#DHUIs  & 49,762 & 	25,080 & 	14,922 & 	9,762 & 	6,836 & 	5,042 \\
&  \textbf{\#${P1}$} (\textit{minCor}: 0.015) &    90,087 & 	63,049 & 	46,155 & 	35,124 & 	27,003 & 	21,197	  \\
\textbf{BMSPOS2}  &  \textbf{\#${P2}$} (\textit{minCor}: 0.020)  & 58,665 & 	42,653 & 	32,439 & 	25,609 & 	20,486 & 	16,700	  	 \\
&\textbf{\#${P3}$} (\textit{minCor}: 0.025) &    41,469 & 	30,896 & 	24,034 & 	19,392 & 	15,879 & 	13,254	  \\
&  \textbf{\#${P4}$} (\textit{minCor}: 0.030)  &   30,857 & 	23,459 & 	18,537 & 	15,182 & 	12,584 & 	10,652  \\
&  \textbf{\#${P5}$} (\textit{minCor}: 0.035) &    23,915 & 	18,519 & 	14,840 & 	12,317 & 	10,327 & 	8,854  \\
&  \textbf{\#${P6}$} (\textit{minCor}: 0.040)  &   19,116 & 	14,970 & 	12,146 & 	10,144 & 	8,563 & 	7,413 	  \\
\hline

&  \#HUIs  &  80,933 & 	39,848 & 	27,839 & 	21,103 & 	16,850 & 	13,722  \\
&  \#DHUIs  & 45,994 & 	24,115 & 	15,961 & 	10,503 & 	6,626 & 	4,434  \\
&  \textbf{\#${P1}$} (\textit{minCor}: 0.12) &    18,129&	15,811&	13,790 &	12,082 &	10,396 &	8,841 	  \\
\textbf{T10I4D100K}  &  \textbf{\#${P2}$} (\textit{minCor}: 0.14)  &   15,260 &	13,283 &	11,664 &	10,287 &	8,884 &		7,594  	 \\
&\textbf{\#${P3}$} (\textit{minCor}: 0.16) &   12,824 & 	11,107 & 	9,772 & 	8,629 & 	7,457 & 	6,424  \\
&  \textbf{\#${P4}$} (\textit{minCor}: 0.18)  &   10,667 & 	9,228 & 	8,132 & 	7,181 & 	6,209 & 	5,380  \\
&  \textbf{\#${P5}$} (\textit{minCor}: 0.20) &   8,917 & 	7,685 & 	6,767 & 	5,986 & 	5,186 & 	4,498  \\
&  \textbf{\#${P6}$} (\textit{minCor}: 0.22)  &  7,425 & 	6,391 & 	5,607 & 	4,952 & 	4,284 & 	3,731  \\
\hline

		\hline
\hline
	\end{tabular}
\end{table*}

\subsection{Data Description and Experimental Setup}

\textbf{Datasets}. Typically e-commerce datasets are proprietary and consequently hard to find among publicly available data. To conduct experiments, we use five publicly available real-world datasets (foodmart, chess\footnote{\url{http://fimi.ua.ac.be/data/}}, mushroom$ ^{3} $) and one synthetic dataset (T10I4D100K) in our experiments.  The characteristics of used datasets are described below in details.

\begin{itemize}
	\item \emph{foodmart}: this dataset is provided with Microsoft SQL Server. It contains 21,556 customer transactions and 1,559 distinct items from an anonymous chain store.
	\item  \emph{chess}: it is a dense dataset since it contains 3,196 transactions with 75 distinct items, and the average transaction length is 36 items. 
	\item  \emph{mushroom}: it has 8,124 transactions with 120 distinct items, and the average transaction length is 23 items. Thus, it is also a dense dataset.

	\item  \emph{retail}: it contains 88,162 transactions with 16,470 distinct items. The average transaction length in retail is 10.3 items. 
	\item  \emph{BMSPOS2}: it has 515,597  transactions with 1,657 distinct items, and the average transaction length is 6.53 items. It collects several years worth of point-of-sale data from a large electronics retailer.
	\item  \emph{T10I4D100K}: this is a synthetic dataset, which has 100,000 transactions with 870 distinct items, and the average transaction length is 10.1 items.
	
\end{itemize}

Note that the foodmart dataset already contains the quantity and a unit profit of each item, while chess and mushroom do not contain the quantitative and profit information. Therefore, we use a simulation method, which is widely adopted in previous studies \cite{tseng2013efficient,liu2012mining,lin2017fdhup}, to generate the quantitative and profit information for each item in the chess and mushroom datasets. For the addressed utility-based mining problem, these used datasets having varied characteristics make the experimental results more convincing and acceptable.

\textbf{Language and experimental environment.} All the algorithms in the experiments were implemented in Java language and performed on a personal ThinkPad T470p computer with an Intel(R) Core(TM) i7-7700HQ CPU @ 2.80 GHz 2.81 GHz, 32 GB of RAM, and with the 64-bit Microsoft Windows 10 operating system.

\textbf{Parameter settings.} It is important to notice that both FHM \cite{fournier2014fhm} and FDHUP \cite{ahmed2011framework} are varied by one parameter $minUtil$, while the CoHUIM and CoUPM algorithms discover the CoHUIs by using two constraints: correlation and utility. Therefore, experiments are conducted on each dataset by varying $minUtil$. In addition, the $ minCor $ is adjusted with six times on each dataset to evaluate the effectiveness of mining patterns. Specifically, as shown in Table \ref{table:patterns1}, the six different $ minCor $ thresholds are respectively set on each data. For instance, in foodmart, $ minCor $ is varying from 0.01 to 0.06, such as 0.01, 0.02, 0.03, 0.04, 0.05, 0.06.

\subsection{Effectiveness Analytics}

The addressed problem aims at computing the satisfiable  correlated and high profitable patterns. Thereby, the derived CoHUIs explicitly includes availability of the correlation and utility contribution. To further investigate the effectiveness of the addressed problem for correlated utility-based pattern mining, we plot in Table \ref{table:patterns1} with the results of different kinds of generated patterns under various parameter settings. Note that the \#HUIs is the number of HUIs discovered by one of traditional HUIM algorithms (e.g., FHM), \#DHUIs is the number of discriminative HUIs discovered by FDHUP, and the \#CoHUIs (it is respectively denoted as \#$P1$ to \#$P6$ under six $minCor$ thresholds) is the number of correlated HUIs discovered by the CoHUIM and CoUPM algorithms. In Table \ref{table:patterns1}, $\alpha$ represents \textit{minUtil}.

As shown in Table \ref{table:patterns1}, it can be clearly observed that the number of CoHUIs is always different from that of \#HUIs and \#DHUIs under various \textit{minCor} and \textit{minUtil} thresholds on all test datasets under all parameter settings. More specifically, both the \textit{minCor} and \textit{minUtil} affect the results of CoHUIs, as shown from \#$P1$ to \#$P6$ on each dataset. In general, the numbers of DHUIs and CoHUIs are always smaller than that of HUIs. These results are reasonable since the DHUIs and CoHUIs are determined with not only the utility constraint, but also the correlation measure. Therefore, to derive desired patterns, more criteria can usually be applied to produce fewer patterns. The difference between \#DHUIs and \#CoHUIs indicates that the addressed problem with \textit{Kulc} measure is more acceptable than  the frequency-affinity-based utility mining framework. It is interesting to observe that the number of DHUIs in chess and mushroom datasets under various $minUtil$ thresholds is close to zero.

In addition, the number of patterns discovered by the designed CoPUM algorithm under six $ minCor $ always has: $ \#P1 \geq  \#P2 \geq \#P3 \geq \#P4 \geq \#P5 \geq \#P6$. When $minUtil$ is fixed on a processed dataset, the larger $ minCor $ is, the smaller the number of derived CoHUIs is. For instance, when $minUtil$ is set as 19\% on chess, \#HUIs is 39,281, \#DHUIs is 0, while the number of CoHUIs is changed from 4,021 to 1,483 (details are $ \#P1$ = 4,021,  $ \#P2$ = 3,378,  $\#P3$ = 2,773,  $ \#P4$ = 2,282, $ \#P5$ = 1,872, and $ \#P6$ = 1,483) when $ minCor $ is varying from 0.74 to 0.79. It indicates that the adopted correlated \textit{Kulc} measure is acceptable and useful to extract non-redundant correlated high-utility patterns from quantitative datasets.

\begin{figure*}[htbp]
	\centering 
	\includegraphics[trim=60 0 65 0,clip,scale=0.63]{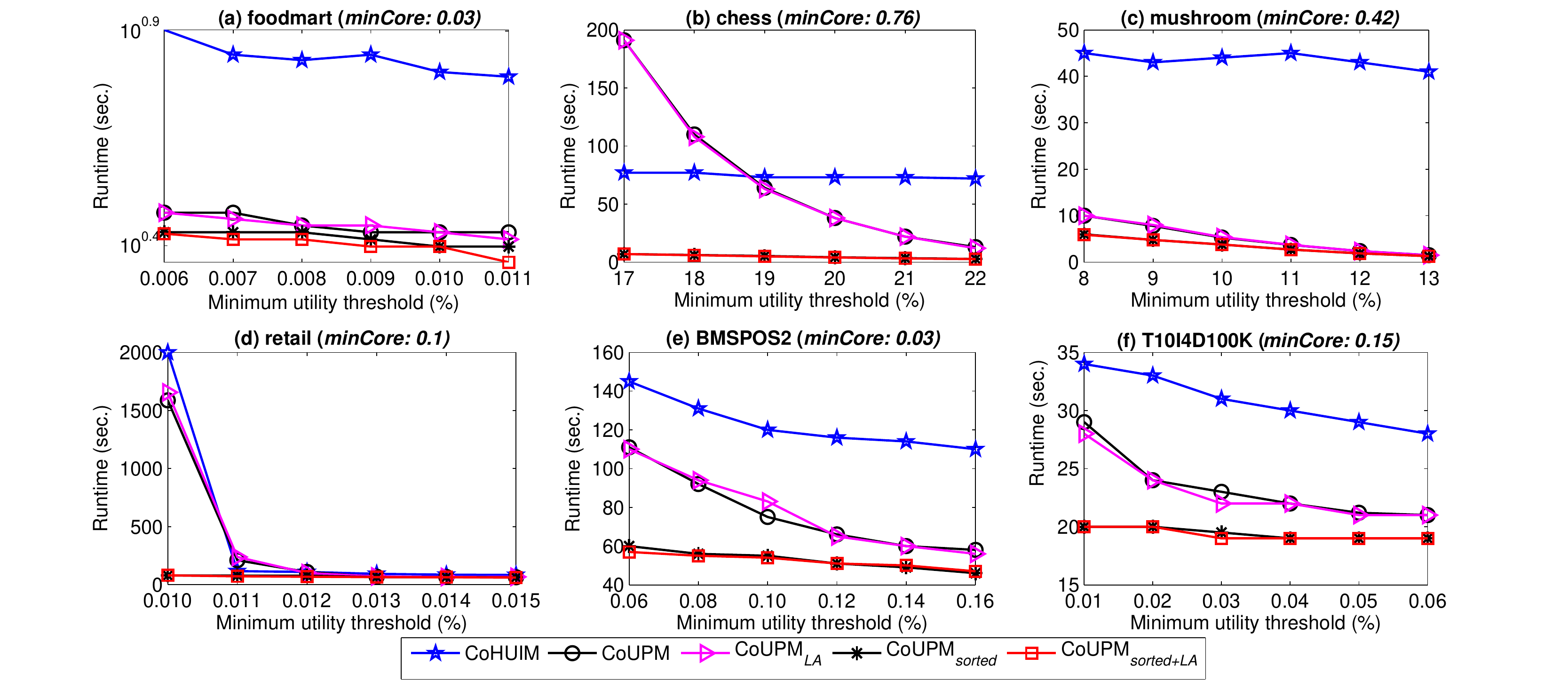}
	\captionsetup{justification=centering}
	\caption{Runtime under various parameters (varying $minUtil$, fix $minCor$).}
	\label{fig:Runtime}	
\end{figure*}

\begin{figure*}[htbp]
	\centering 
	\includegraphics[trim=60 0 65 0,clip,scale=0.63]{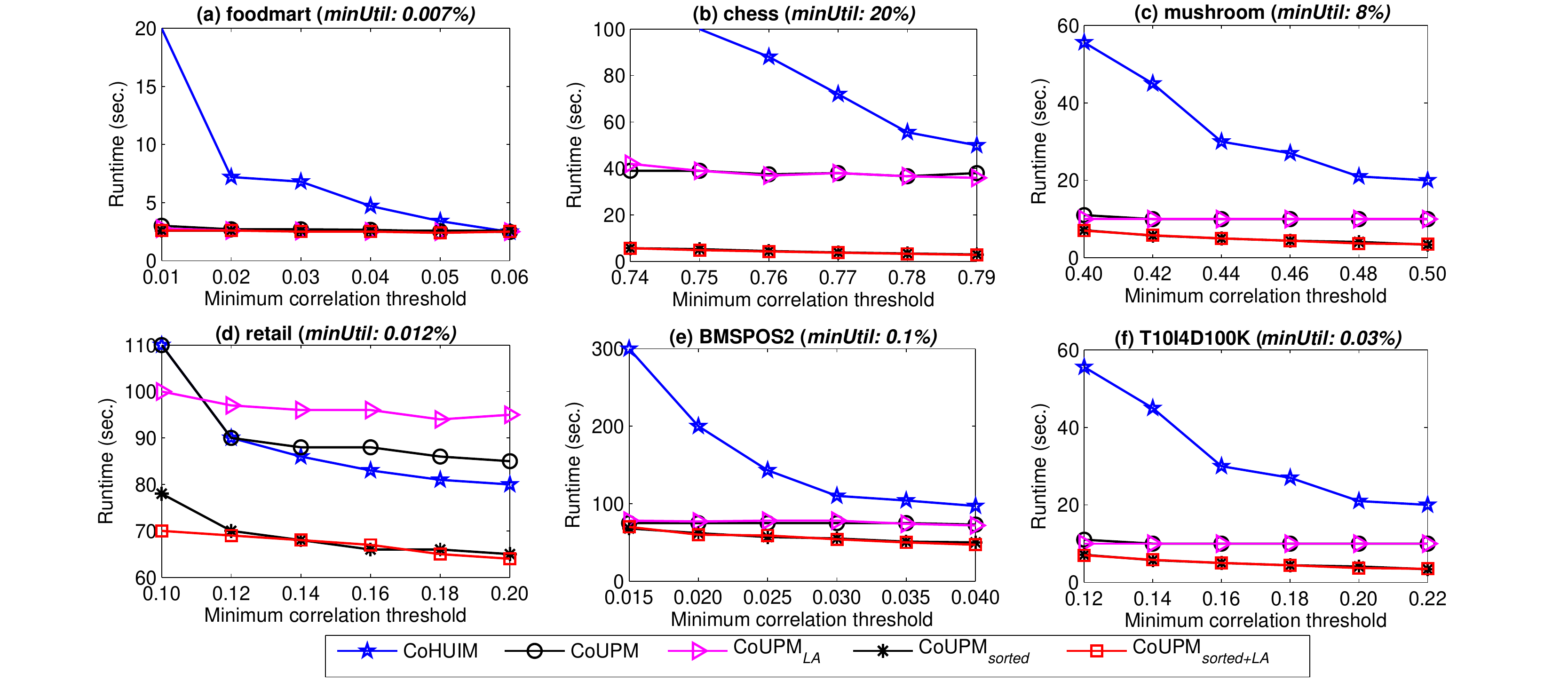}
	\captionsetup{justification=centering}
	\caption{Runtime under various parameters (fix $minUtil$, varying $minCor$).}
	\label{fig:Runtime2}	
\end{figure*}

\subsection{Efficiency Analytics}

From Table \ref{table:patterns1}, we can observe that the mining results with the influence of \textit{minCor} threshold and \textit{minUtil} threshold. In this subsection, we continue to perform the evaluation of efficiency in terms of running time.  To make fair comparison, we use the same parameter settings which are tested in Table \ref{table:patterns1}. Both the minimum correlation threshold and the minimum utility threshold are used to evaluate the efficiency. We investigate the processing time of CoHUIM, CoUPM, and its three improved variants in six real datasets by varying $minUtil$ and $minCore$. When varying one threshold, another one is fixed on each dataset. The results of total execution time of the four variants are presented in Figure \ref{fig:Runtime} and Figure \ref{fig:Runtime2}, respectively. In particular,  CoUPM$_{sorted}$ means the designed CoUPM algorithm with Strategy 1 which utilizes the \textit{sorted downward closure} property of \textit{Kulc} measure), while CoUPM is executed without using Strategy 1.

Firstly, CoUPM with or without Strategy 1 consistently outperforms the state-of-the-art CoHUIM approach, even up to 3 orders of magnitude. In particular, CoUPM$_{sorted}$ outperforms CoUPM in most cases under all parameter settings. For example, in the case in Figure \ref{fig:Runtime}(e), we can obviously observe the difference of the runtime  between CoUPM$_{sorted}$ and CoUPM. When $minUtil$ is set to 20\% on chess dataset, the runtime of CoUPM$_{sorted}$ always closes to 4 seconds, while CoUPM approximately has its processing time as 40 seconds. This difference also can be observed from the other datasets. This observation indicates that the \textit{sorted downward closure} property of \textit{Kulc} measure plays an active role in pruning the search space of the correlation-based CoUPM algorithm.

Based on the observation of runtime between CoUPM and CoUPM$_{LA}$, it indicates that the LA-Prune strategy also plays an active role in filtering the unpromising patterns in some cases. To summarize, the improved algorithms which utilize the powerful pruning strategies always have the best performance compared to the baseline CoUPM algorithm, as well as the CoHUIM algorithm.

It is important to notice that the projection-based CoHUIM algorithm may be very time-consuming on low thresholds or dense datasets. And this computational efficiency problem might be more easily happened in dense datasets, which can been seen in the view of Figure \ref{fig:Runtime}(b), Figure \ref{fig:Runtime}(c), Figure \ref{fig:Runtime2}(b), and Figure \ref{fig:Runtime2}(c), respectively. Overall, the proposed CoUPM algorithm significantly has better performance than the state-of-the-art CoHUIM algorithm in terms of running time and memory consumption. On dense datasets, i.e., chess and mushroom, the consumed memory of CoHUIM is very huge and can up to 50 times than that of CoUPM.

\begin{figure*}[htbp]
	\centering 
	\includegraphics[trim=20 0 50 0,clip,scale=0.56]{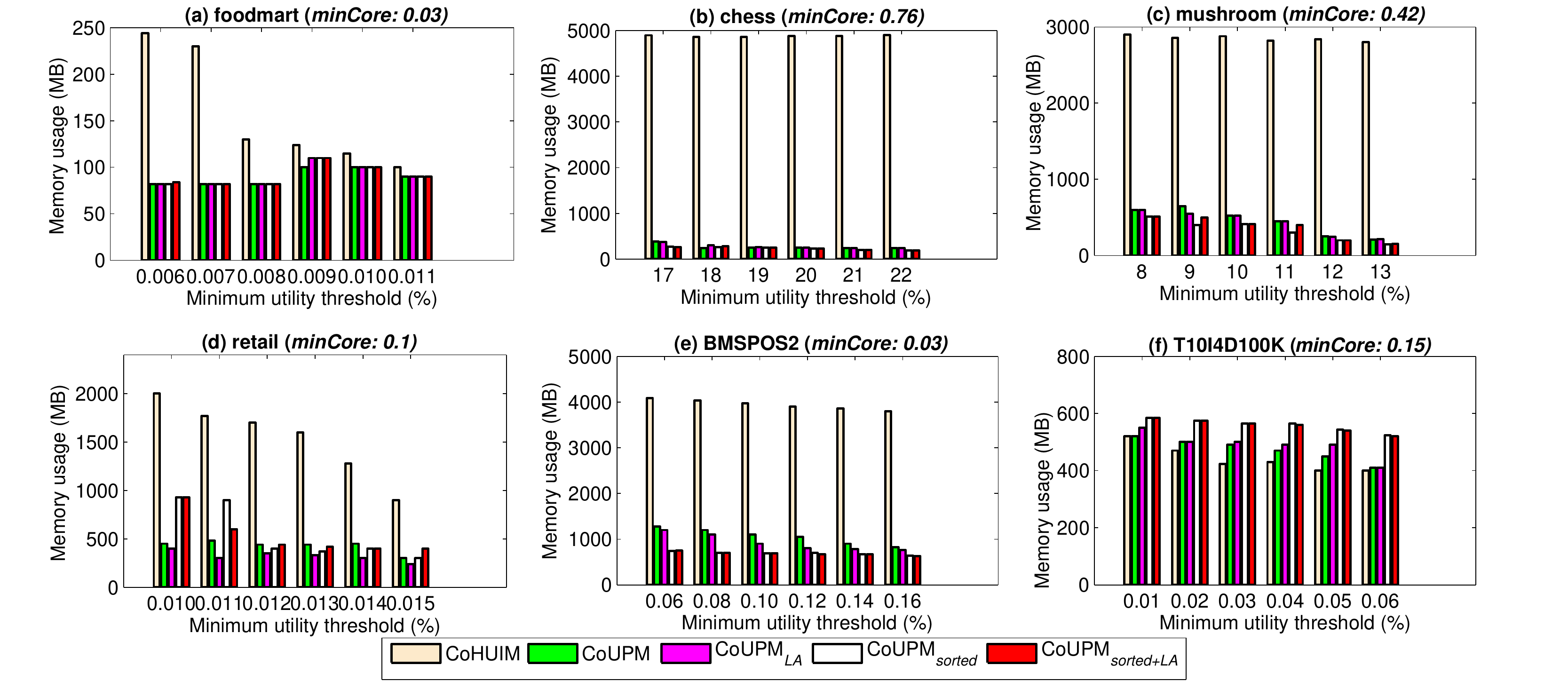}
	\captionsetup{justification=centering}
	\caption{Memory usage under various parameters (varying  $minUtil$, fix $minCor$).}
	\label{fig:Memory}	
\end{figure*}

\subsection{Memory Evaluation}

In this subsection, we continue to evaluate the memory consumption of the compared algorithms. Results of the peak memory usage of CoHUIM and different variants of CoUPM on the all test datasets  with the same parameter settings in Figure \ref{fig:Runtime} and Figure \ref{fig:Runtime2} are shown in Figure \ref{fig:Memory} and Figure \ref{fig:Memory2}, respectively. Note that we use the Java API to calculate the peak memory consumption of each compared algorithm during the whole mining process.

As we can see, all the revised utility-list-based models perform significantly better than the projection-based CoHUIM algorithm, demonstrating the suitability of these models for dense datasets or large-scale datasets. For example, as shown in Figure \ref{fig:Memory}, the peak memory consumption for CoUPM is significantly less than that of CoHUIM. In addition, the improved variants, e.g., CoUPM$_{sorted}$, CoUPM$_{LA}$ and CoUPM$_{sorted+LA}$, consume less memory than the baseline CoUPM algorithm that only adopts the pruning Strategy 2.

The peak memory consumptions under various values of parameters ($minUtil$ and $minCor$) are shown in Figure \ref{fig:Memory} and Figure \ref{fig:Memory2}, respectively. Note that the \textit{y}-axis shows the peak memory consumption of the whole mining process corresponding to the choice of minimum utility threshold (\textit{minUtil}) and minimum correlation threshold (\textit{minCore}). As what can be seen, the proposed CoUPM model with several pruning strategies outperforms CoHUIM for all parameter settings. As mentioned previously, the advantage of CoUPM is that it is able to early filter a large amount of unpromising patterns by building the initial revised utility-lists. As the size of explored pattern increases, the revised utility-list size decreases, thus CoUPM exceeds the available main memory and its overall execution time decreases significantly. For instance, at mushroom (\textit{minCore} = 0.42 and \textit{minUtil} =10\% at Figure \ref{fig:Memory}), CoUPM has a peak memory consumption of 650 MB and requires 6s to discover the required information. For the same minimum thresholds, CoHUIM has a peak memory consumption of 2,900 MB and requires 43s for the mining task. This corresponds to a speed-up of 4.5x in memory, and speed-up of 7x in execution time. For dense data sizes (e.g., chess, mushroom), the speed-up increases further.

\begin{figure*}[htbp]
	\centering 
	\includegraphics[trim=20 0 40 0,clip,scale=0.56]{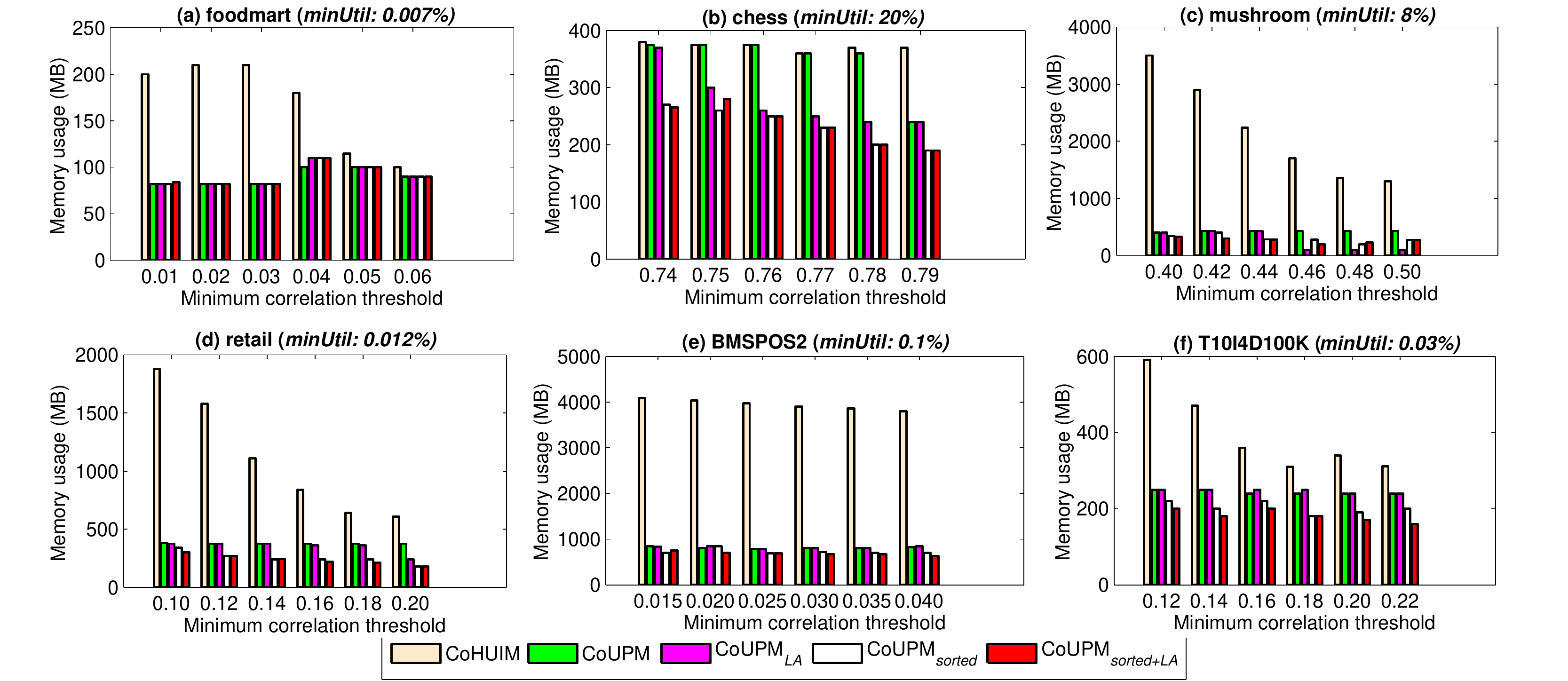}
	\captionsetup{justification=centering}
	\caption{Memory usage under various parameters (fix $minUtil$, varying $minCor$).}
	\label{fig:Memory2}	
\end{figure*}

\subsection{Summary and Discussion}

For the proposed CoUPM, we have the following observations: (1) Best performance is achieved when all data structures (revised utility-lists) fit in main memory. (2) The performance degrades but still remains acceptable while dealing with dense dataset. (3) The low performance of projection-based CoHUIM model may be related to the huge memory consumption which is quite important in utility mining. In summary, we have the following observations of the results.

\begin{itemize}
	\item First, the filtered estimation of upper bound on utility takes a positive role in early pruning the unpromising patterns based on the revised utility-list to store the compact but complete information. 
	
	\item Second, the designed CoUPM algorithm makes use of the compact data structure named revised utility-list. Thus, it can efficiently hold the ``mining during the constructing" property, and the real search space and memory cost can be significantly reduced. On the contrast, the projection CoHUIM approach which recursively projects the sub-databases for next iteration may easily encounter a huge of memory cost, especially on dense datasets. 
	
	\item Third, by utilizing the proposed pruning strategies with the properties of correlation and  upper bound on utility, the search space and memory cost of the CoUPM algorithm is further reduced. The worse performance of the CoHUIM algorithm is caused by the candidate generation-and-test mechanism. 
	
	\item In general, the upper bond on utility used in CoHUIM is not tight enough, and a huge number of candidates are required to be generated although the \textit{sorted downward closure} property of \textit{Kulc} is adopted in CoHUIM to prune the candidates in the search space.  
\end{itemize}

\section{Conclusion and Future Work}
\label{sec:conclusion}

In this paper, we have presented an efficient utility mining framework named CoUPM for discovering non-redundant correlated high-utility patterns from quantitative databases. It studies the problem of utility-based pattern mining by measuring both correlation and availability of utility.  Based on the revised utility-list, CoUPM does not need to scan the database with multiple times. It relies on several pruning strategies, which utilize the sorted downward closure of \textit{Kulc} and upper bound on utility based on the concept of \textit{remaining utility}. Moreover, CoUPM can directly discover the desired patterns from the quantitative databases by avoiding performing costly intersection operations of revised utility-lists. The extensive performance on several real-world datasets demonstrates the effectiveness and efficiency of the CoUPM algorithm. 

For the future work, we plan to improve the mining efficiency by developing a new data structure instead of using the utility-list for the addressed problem. Secondly, we would focus on other practical effectiveness issues of utility mining. For example, we would like to conduct further research of the proposed model to deal with the dynamic utility mining \cite{lin2015fast,2gan2018survey},  utility mining on uncertain data \cite{lin2016efficient}, and privacy issue \cite{gan2018privacy}. Lastly, it is also interesting to take the other interesting extensions and applications into account for our future studies.


\section*{Acknowledgment}

This work was partially supported by the Shenzhen Technical Project under JCYJ 20170307151733005 and KQJSCX  20170726103424709.  Specifically, Wensheng Gan was supported by the CSC (China Scholarship Council) Program during the study at University of Illinois at Chicago, IL, USA.

\section*{References}

\bibliography{main} 

\end{document}